\documentclass[letterpaper, 10 pt, conference]{ieeeconf}  

\IEEEoverridecommandlockouts                              

\usepackage{graphicx} 
\usepackage{amsmath} 
\usepackage{amssymb}  

\usepackage{amsthm}
\usepackage{cite}
\usepackage{hyperref}
\usepackage{accents}
\usepackage{centernot}
\usepackage{color}
\usepackage{algorithm}
\usepackage[noend]{algpseudocode}
\usepackage[font=scriptsize]{caption}
\usepackage{subcaption} 

\renewcommand{\baselinestretch}{0.9}

\newtheorem{ass}{Assumption}
\newtheorem{thm}{Theorem}
\newtheorem{crl}{Corollary}

\title{\LARGE \bf
Reduced Order Model Predictive Control For Setpoint Tracking
}

\author{Joseph Lorenzetti, Benoit Landry, Sumeet Singh, Marco Pavone	
\thanks{The authors are with the Department of Aeronautics and Astronautics, Stanford University, Stanford CA. Emails: \{jlorenze, blandry, ssingh19, pavone\}@stanford.edu.}
\thanks{This work was supported by the Office of Naval Research  (Grant N00014-17-1-2749). Joseph Lorenzetti is supported by the Department of Defense (DoD) through the National Defense Science and Engineering Fellowship (NDSEG) Program.}
}

\begin{document}

\maketitle
\thispagestyle{empty}
\pagestyle{empty}

\begin{abstract}
Despite the success of model predictive control (MPC), its application to high-dimensional systems, such as flexible structures and coupled fluid/rigid-body systems, remains a largely open challenge due to excessive computational complexity. A promising solution approach is to leverage reduced order models for designing the model predictive controller. In this paper we present a reduced order MPC scheme that enables setpoint tracking while robustly guaranteeing constraint satisfaction for linear, discrete, time-invariant systems. Setpoint tracking is enabled by designing the MPC cost function to account for the steady-state error between the full and reduced order models. Robust constraint satisfaction is accomplished by solving (offline) a set of linear programs to provide bounds on the errors due to bounded disturbances, state estimation, and model approximation. The approach is validated on a synthetic system as well as a high-dimensional linear model of a flexible rod, obtained using finite element methods.
\end{abstract}

\section{Introduction}
Model predictive control (MPC) is an advanced control technique that entails optimizing predicted future behavior in a receding horizon fashion. This is accomplished by solving, at each time step, an optimal control problem that includes a model of the system. The approach is especially suited for control of \textit{constrained} systems due to the ability to incorporate state and control constraints within the optimal control problems.

MPC grew in popularity for process control in industrial settings, where safety is critical and system dynamics are typically slow. As more powerful computational resources have become available, the appeal of MPC has spread to control systems with faster dynamics such as robotics and autonomous vehicles. For example, MPC has been applied to the control of aerospace systems \cite{ErenPrachEtAl2017}, robotic manipulators \cite{PoignetGautier2000}, ground vehicles \cite{BealGerdes2013}, and many others. However, there are still many interesting systems for which the use of MPC is computationally prohibitive. These systems typically have very high-dimensional state-space representations, such as from discrete approximations of infinite-dimensional systems. Infinite-dimensional systems, commonly characterized by partial differential equations, arise in many real-world robotic applications. Examples include soft robotics \cite{GouryDuriez2018}, flexible structures and robotic manipulators \cite{RaoPanEtAl1990}, and autonomous systems with coupled fluid/rigid-body dynamics due to aerodynamics \cite{AmsallemDeolalikarEtAl2013} or fluid sloshing \cite{ArdakaniBridges2011}.

One approach to reduce the computational burden of using MPC with high-dimensional systems is to solve the optimal control problem using reduced order models, which are lower order approximations of the original system dynamics \cite{Antoulas2005}. While this approach can provide computational tractability, it introduces model approximation error, which in turn can induce setpoint tracking bias and/or system constraint violations. 
In this work we address the problems of setpoint tracking and robust constraint satisfaction when using reduced order model predictive control (ROMPC), where the robustness is with respect to bounded disturbances, state estimation error, \emph{and} model reduction error.

{\em Related Work:}
The use of MPC to regulate constrained systems has been studied extensively over the last several decades (see \cite{RawlingsMayne2017} and references therein). MPC methods for tracking have also been studied; for example in \cite{MaederBorrelliEtAl2009} an MPC algorithm is developed that guarantees offset-free tracking of asymptotically constant references at steady state, even with model mismatch. Additionally, setpoint tracking is addressed in \cite{AlvaradoLimonEtAl2007} with a robust output MPC scheme, albeit without considering modeling errors.

Reduced order modeling \cite{Antoulas2005} is also a well established field. While often used for simulation, reduced order models have also been applied to control problems \cite{AmsallemDeolalikarEtAl2013, AstridHuismanEtAl2002}  using ROMPC~\cite{HovlandWillcoxEtAl2006, HovlandGravdahlEtAl2008, HovlandLovaasEtAl2008, SopasakisBernardiniEtAl2013, LoehningRebleEtAl2014, KoegelFindeisen2015b}. In \cite{SopasakisBernardiniEtAl2013}, guarantees on constraint satisfaction are given by considering the neglected dynamics as a bounded disturbance. This is extended in \cite{LoehningRebleEtAl2014} to guarantee asymptotic stability to the origin by using an error bounding system. However both \cite{SopasakisBernardiniEtAl2013} and \cite{LoehningRebleEtAl2014} require knowledge of the state of the full order system to include feedback of the model errors. A robust output MPC scheme using reduced order models was then introduced in \cite{KoegelFindeisen2015b} which uses a tube-based approach and guarantees constraint satisfaction by computing error bounds \textit{a priori}. However, the existing ROMPC literature has not yet addressed the problem of setpoint tracking, which is the focus of this work.

{\em Statement of Contributions:}
In this paper we introduce a method for setpoint tracking of linear, discrete, time-invariant \textit{high-dimensional} systems using a \textit{reduced order} model predictive control scheme (ROMPC). Setpoint tracking is accomplished by computing target state and control values for the ROMPC problem that account for the model reduction error. In addition, we introduce a method for computing error bounds that enables robust constraint satisfaction with respect to bounded disturbances, state estimation error, and model reduction error. Our method for computing the error bounds is inspired by \cite{KoegelFindeisen2015b}, but is less conservative. The proposed method is validated on a synthetic example and on an engineering-based problem where we control a flexible structure with a high-dimensional linear model generated using a finite element method.

{\em Organization:} First, Section \ref{sec:formulation} formulates the control task. Next, Section \ref{sec:control} lays out the ROMPC solution methodology. In Section \ref{sec:stability} we discuss stability of the ROMPC scheme. Section \ref{sec:errorbounds} discusses the derivation of the ROMPC constraints guaranteeing robust constraint satisfaction. Then, in Section \ref{sec:tracking} we characterize the setpoint tracking performance of our method. Section \ref{sec:examples} provides simulation results of the proposed approach. Finally, we draw our conclusions in Section \ref{sec:conclusion} and provide avenues for future work.

\section{Problem Formulation}\label{sec:formulation}
In this work we consider linear, discrete, time-invariant systems and assume bounded, additive disturbances on the process and measurements. We start with the description of the high-dimensional system that we wish to control, hereby referred to as the full order system.

\subsection{Full Order System Model} \label{fomsection}
The model for the full order system is given by:
\begin{equation} \label{fom}
\begin{split}
x^f_{k+1} &= A^fx^f_k + B^fu_k + w^f_k, \\
y_k &= C^f x^f_k + v_k, \quad z_k = H^f x^f_k,\\
\end{split}
\end{equation}
where $x^f \in \mathbb{R}^{n^f}$ is the full dimensional state, $u \in \mathbb{R}^m$ is the control input, $y \in \mathbb{R}^p$ is the measured output, $z \in \mathbb{R}^o$ are the performance variables, $w^f \in \mathbb{R}^{n^f}$ is the process noise, and $v \in \mathbb{R}^p$ is the measurement noise. This model could arise, for example, from a high-dimensional finite approximation of an infinite-dimensional model (e.g., by discretizing partial differential equations)  \cite{SopasakisBernardiniEtAl2013}. Constraints on the performance variables, control, process, and measurement noise are defined by bounded convex polytopes:
\begin{align}
z_k \in \mathcal{Z}, \:\: u_k \in \mathcal{U}, \label{const} \\
w^f_k \in \mathcal{W}, \:\: v_k \in \mathcal{V}. \label{noise}
\end{align}
It is assumed that the pair $(A^f, B^f)$ is stabilizable. 

We define the tracking performance variables $z_k^r \in \mathbb{R}^t$ as $z_k^r := Tz_k$, where the matrix $T \in \mathbb{R}^{t \times o}$ is defined by taking the $t$ rows of the identity matrix $I_o$ corresponding to the indices of the performance variables we wish to track. The objective then is to control the system (\ref{fom}) such that the tracking performance variables $z^r_k$ track a desired setpoint (i.e., constant signal) $r$. In the absence of noise, perfectly tracking an arbitrary setpoint requires the system  (\ref{fom}) to reach a steady state, defined by $x^f_{\infty}$ and $u_{\infty}$. These quantities are given as the solution to the following linear system:
\begin{equation} \label{steadystate}
S_f \begin{bmatrix}
x^f_{\infty} \\ u_{\infty}
\end{bmatrix} = 
\begin{bmatrix}
0 \\ r
\end{bmatrix}, \quad S_f=\begin{bmatrix}
A^f - I & B^f \\ TH^f & 0
\end{bmatrix}.
\end{equation}
To ensure that~\eqref{steadystate} has a unique solution, we assume that the matrix $S_f$ is square (i.e., number of tracking variables $t$ is equal to the number of control inputs $m$), and full rank.

\subsection{Setpoint Tracking}
The setpoint tracking problem for a \textit{constrained} system is often addressed through the use of model predictive control. Let $\hat{x}^f_k$ denote the state estimate as computed via an observer based on the system (\ref{fom}). A standard approach, as proposed in \cite{MayneRakovicEtAl2006}, entails solving the optimal control problem
\begin{equation} \label{highordercontrol}
\begin{split}
\underset{\mathbf{x^f_k}, \mathbf{u_k}}{\text{min.}} \:\:& V(\mathbf{x^f_k}, \mathbf{u_k}, x^f_{\infty}, u_{\infty}), \\
\text{subject to} \:\: & x^f_{i+1|k} = A^fx^f_{i|k} + B^fu_{i|k},  \\
& H^fx^f_{i|k} \in \mathcal{Z}, \:\: u_{i|k} \in \mathcal{U}, \\
& x^f_{k+N|k} \in \mathcal{X}_f, \quad \hat{x}^f_k \in x^f_{k|k} \oplus \mathcal{D}^f, \\
\end{split}
\end{equation}
over horizon $N$, where the decision variables are $\mathbf{x^f_k} := [x^f_{k|k},\dots,x^f_{k+N|k}]$ and $\mathbf{u_k} := [u_{k|k},\dots,u_{k+N-1|k}]$,  $i=k,\dots,k+N-1$, the sets $\mathcal{X}_f$ and $\mathcal{D}^f$ are discussed below, and $\oplus$ denotes Minkowski addition. The objective function is given by
\begin{equation}
\begin{split}
&V(\mathbf{x^f_k}, \mathbf{u_k}, x^f_{\infty}, u_{\infty}) = \\
&||\delta x^f_{k+N|k}||^2_P + \sum_{j=k}^{k+N-1}||\delta x^f_{j|k}||^2_Q + ||\delta u_{j|k}||^2_R,
\end{split}
\end{equation}
where $\delta x^f = x^f - x^f_{\infty}$, $\delta u = u - u_{\infty}$, and $P$, $Q$, and $R$ are positive definite cost matrices. The values of $x^f_{\infty}$ and $u_{\infty}$ are given by a solution to  the linear system (\ref{steadystate}). The terminal set, $\mathcal{X}_f$, along with the terminal cost $P$, are designed to provide stability properties to the algorithm. Lastly, the set $\mathcal{D}^f$ is a bound on the error $\hat{x}^f_k - x^f_k$. Further details, including computation of the sets $\mathcal{D}^f$ and $\mathcal{X}_f$ are discussed in \cite{MayneRakovicEtAl2006}.

As the dimension of the system (\ref{fom}) increases, so does the size of the optimal control problem (\ref{highordercontrol}). Therefore, when the system is high-dimensional, the method proposed in \cite{MayneRakovicEtAl2006} may be too computationally limiting to use for real-time control. To address these situations, we propose to leverage recent work in MPC based on \textit{reduced order models} to solve the setpoint tracking problem.

\subsection{Reduced Order Model}
The \textit{nominal} reduced order model for the full order system (\ref{fom}) is defined by
\begin{equation} \label{rom}
\begin{split}
\bar{x}_{k+1} &= A\bar{x}_k + B\bar{u}_k, \\
\bar{y}_k &= C \bar{x}_k, \quad \bar{z}_k = H \bar{x}_k, \\
\end{split}
\end{equation}
where $\bar{x} \in \mathbb{R}^n$ is the nominal reduced order state, $\bar{u}\in \mathbb{R}^m$ is the control input to the reduced order system,  $\bar{y} \in \mathbb{R}^p$ is the nominal output based on the reduced model, and $\bar{z} \in \mathbb{R}^o$ is the expected performance variable based on the reduced model. The reduced order system matrices are $A\in \mathbb{R}^{n \times n}$, $B\in \mathbb{R}^{n \times m}$, $C\in \mathbb{R}^{p \times n}$, and $H\in \mathbb{R}^{o \times n}$, where it is often the case that $n \ll n^f$. Note that the control input $\bar{u}$, output $\bar{y}$, and performance variable $\bar{z}$ in the reduced order model are of the same dimensionality as the input $u$ and output $y$ of the full order system. In this work, no assumptions are made on the specific type of model order reduction techniques used to obtain (\ref{rom}). However, it is assumed that the pair $(A, B)$ is controllable and that the pair $(A, C)$ is observable.

To summarize, the full order model (\ref{fom}) is the system to be controlled such the tracking performance variables, $z^k_r$, track a desired setpoint. The reduced order model (\ref{rom}) is the system used to \textit{design} the model predictive controller.

\section{Reduced Order MPC (ROMPC)}\label{sec:control}
In this section, we present the key blocks of our control methodology: a {\em reduced order} observer (Section \ref{subsec:ob}) and a {\em reduced order} control law (Section \ref{subsec:co}), and the corresponding ROMPC problem (Section \ref{subsec:rompc}).

\subsection{Reduced Order Observer}\label{subsec:ob}
The state of the full order system is not assumed to be known, and estimating the full order state $x^f$ requires a high-dimensional observer. Instead we use the reduced order model (\ref{rom}) in a Luenberger observer to estimate the nominal \textit{reduced order} state $\bar{x}$ 
\begin{equation} \label{estimator}
\hat{x}_{k+1} = A\hat{x}_k + Bu_k + L(y_k - C\hat{x}_k),
\end{equation}
where $\hat{x}_k$ is the reduced order state estimate, and $y_k$, $u_k$ are the measurement and control from the \textit{full order system} (\ref{fom}). The gain $L$ is chosen such that $A-LC$ is Schur stable.

\subsection{Reduced Order Controller}\label{subsec:co}
The overall control strategy is to recursively generate a nominal control trajectory over a horizon $N$, namely $\mathbf{\bar{u}^*_k} := [\bar{u}^*_k,\dots,\bar{u}^*_{k+N-1}]$, along with a nominal state trajectory, namely  $\mathbf{\bar{x}^*_k} := [\bar{x}^*_k,\dots,\bar{x}^*_{k+N}]$, by solving a reduced order MPC problem (detailed in Section \ref{subsec:rompc}). The nominal state trajectory is then tracked using a linear feedback controller based on the reduced order state estimate $\hat{x}_k$. The reduced order controller is
\begin{equation} \label{controller}
u_k = \bar{u}^*_k + K(\hat{x}_k - \bar{x}^*_k).
\end{equation}
The gain matrix $K$ is computed such that $A+BK$ is Schur stable.

\subsection{ROMPC Problem}\label{subsec:rompc}
The nominal state and control trajectories ($\mathbf{\bar{x}^*_k}, \mathbf{\bar{u}^*_k}$) are computed by solving a {\em reduced order MPC problem} (i.e., an MPC problem based on the reduced order model (\ref{rom})). While this improves the computational performance, it also adds several challenges to the design of the optimization problem. First, setpoint tracking is non-trivial. Target values for the reduced order state and control must be computed that account for both the observer and model reduction error. Second, the \textit{reduced order} MPC problem must guarantee that the \textit{full order} system constraints (\ref{const}) are satisfied.

Specifically, the nominal state and control trajectories ($\mathbf{\bar{x}^*_k}, \mathbf{\bar{u}^*_k}$) at time step $k$ are computed by solving the optimal control problem: 
\begin{equation} \label{optimalcontrol}
\begin{split}
V^*_k(\bar{x}_k, \bar{x}_{\infty}, \bar{u}_{\infty}) := & \underset{\mathbf{\bar{x}_k}, \mathbf{\bar{u}_k}}{\text{min.}} \:\: V(\mathbf{\bar{x}_k}, \mathbf{\bar{u}_k}, \bar{x}_{\infty}, \bar{u}_{\infty}), \\
\text{subject to} \:\: & \bar{x}_{i+1|k} = A\bar{x}_{i|k} + B\bar{u}_{i|k},  \\
& H\bar{x}_{i|k} \in \bar{\mathcal{Z}}, \:\: \bar{u}_{i|k} \in \bar{\mathcal{U}}, \\
& \bar{x}_{k+N|k} \in \bar{\mathcal{X}}_f,\: \bar{x}_{k|k} = \bar{x}_k, \\
\end{split}
\end{equation}
where $i=k,\dots,k+N-1$. The decision variables are the nominal states $\mathbf{\bar{x}_k} = [\bar{x}_{k|k},\dots,\bar{x}_{k+N|k}]$ and the nominal control inputs $\mathbf{\bar{u}_k} = [\bar{u}_{k|k},\dots,\bar{u}_{k+N-1|k}]$, and the first argument $\bar{x}_k$ is the current nominal reduced order state whose dynamics are given by (\ref{rom}). The objective function is quadratic and given by
\begin{equation} \label{cost}
\begin{split}
&V(\mathbf{\bar{x}_k}, \mathbf{\bar{u}_k}, \bar{x}_{\infty}, \bar{u}_{\infty}) = \\
&||\delta \bar{x}_{k+N|k}||^2_P + \sum_{j=k}^{k+N-1}||\delta \bar{x}_{j|k}||^2_Q + ||\delta \bar{u}_{j|k}||^2_R,
\end{split}
\end{equation}
where $\delta \bar{x} = \bar{x} - \bar{x}_{\infty}$, $\delta \bar{u} = \bar{u} - \bar{u}_{\infty}$, and the penalty matrices $P$, $Q$, and $R$ are positive definite. The nominal reduced order performance variables $\bar{z} = H\bar{x}$ and control $\bar{u}$ in (\ref{optimalcontrol}) are constrained to lie in the sets $\bar{\mathcal{Z}}$ and $\bar{\mathcal{U}}$, respectively. These sets are tightened versions of the constraints $\mathcal{Z}$ and $\mathcal{U}$. The use of tightened constraint sets is required so that the controller (\ref{controller}) can robustly ensure constraint satisfaction for the full order system in the presence of bounded disturbances, estimation error, and model reduction error. The computation of the sets $\bar{\mathcal{Z}}$ and $\bar{\mathcal{U}}$ is discussed in Section \ref{sec:errorbounds}. The second and third arguments to the optimal control problem, $\bar{x}_{\infty}$, $\bar{u}_{\infty}$, are the nominal reduced order system target values. These values are chosen to enable setpoint tracking by accounting for model reduction errors, and their computation is discussed in Section \ref{trackingcond}.

To guarantee robust constraint satisfaction and enable setpoint tracking, the ROMPC problem must also be recursively feasible and stable.
Recursive feasibility and stability of the ROMPC scheme is guaranteed through the design of the terminal cost matrix $P$ and the terminal set $\bar{\mathcal{X}}_f$ (as is typical in MPC schemes \cite{RawlingsMayne2017}) and is discussed next in Section \ref{sec:stability}.

\section{ROMPC Stability}\label{sec:stability}
Having introduced the reduced order control problem, we now discuss the design of the ROMPC problem (\ref{optimalcontrol}) to guarantee closed-loop stability (i.e., it is recursively feasible and has guaranteed convergence). Specifically, we use the well-known approach presented in \cite{RawlingsMayne2017} to design a terminal controller $\kappa(\bar{x})$, a terminal set $\bar{\mathcal{X}}_f$, and a terminal cost $V_f(\bar{x}) = ||\bar{x} - \bar{x}_{\infty}||^2_P$ that satisfy the following two conditions:
\begin{equation} \label{fc}
\begin{split}
A\bar{x} + B\kappa(\bar{x}) \subseteq \bar{\mathcal{X}}_f,& \:\: \forall \bar{x} \in \bar{\mathcal{X}}_f, \\
H\bar{\mathcal{X}}_f \subseteq \bar{\mathcal{Z}}, \:\: \kappa(\bar{x}) \subseteq \bar{\mathcal{U}},& \:\: \forall \bar{x} \in \bar{\mathcal{X}}_f, \text{ and} \\
\end{split}
\end{equation}
\begin{equation} \label{sc}
V_f(\bar{x}_{k+1}) + l(\bar{x}_k, \kappa(\bar{x}_k)) \leq V_f(\bar{x}_k), \:\: \forall \bar{x}_k \in \bar{\mathcal{X}}_f,
\end{equation}
where $l(\bar{x}, \bar{u}) = ||\bar{x} - \bar{x}_{\infty}||^2_Q + ||\bar{u} - \bar{u}_{\infty}||^2_R$ is the stage cost.

To guarantee recursive feasibility, the first condition, (\ref{fc}), requires that the terminal set is positively invariant for the nominal dynamics under the terminal controller and that the resulting terminal control is admissible over the terminal set. The second condition, (\ref{sc}), is used to guarantee convergence.
We now discuss the design of the terminal controller $\kappa(\bar{x})$, the terminal set $\bar{\mathcal{X}}_f$, and the terminal cost $V_f(\bar{x})$ that together satisfy these conditions.

Consider the unconstrained infinite-horizon LQR problem with cost weights $Q,R$ for the error system:
\begin{equation} \label{terminaldynamics}
\delta \bar{x}_{k+1} = A\delta \bar{x}_k + B\delta \bar{u}_k,
\end{equation}
where $\delta \bar{x} = \bar{x}-\bar{x}_{\infty}$ and $\delta \bar{u}_k = \bar{u}_k - \bar{u}_{\infty}$. Let $P$ be the solution of the associated discrete algebraic Riccati equation, and set the terminal cost to be $V_f(\bar{x}) = (\bar{x}-\bar{x}_{\infty})^TP(\bar{x}-\bar{x}_{\infty}) = \delta \bar{x}^TP\delta \bar{x}$. Let $K_f$ be the optimal control gain for the associated infinite horizon LQR problem, and parameterize the terminal controller as:
\begin{equation} \label{termcont}
\kappa(\bar{x}) = K_f \delta \bar{x} + \bar{u}_{\infty}, \:\: \forall \bar{x} \in \bar{\mathcal{X}}_f.
\end{equation}
To address condition~\eqref{fc}, we introduce the set $\Delta$ such that $\bar{\mathcal{X}}_f := \Delta \oplus \{\bar{x}_{\infty}\}$. The set $\bar{\mathcal{X}}_f$ that satisfies the stabilizing condition (\ref{fc}) under the terminal control (\ref{termcont}) is then found by computing a set $\Delta$ (see \cite{RawlingsMayne2017}) that satisfies:
\begin{equation*}
\begin{split}
(A_{K_f})\Delta \subseteq \Delta, \:\: H\Delta \subseteq \bar{\mathcal{Z}} \ominus \{H\bar{x}_{\infty}\},\:\:  K_f \Delta \subseteq \bar{\mathcal{U}} \ominus \{\bar{u}_{\infty}\},
\end{split}
\end{equation*}
where $A_{K_f} = A+BK_f$. These conditions are simply the conditions (\ref{fc}) with the substitution of the terminal controller and the definition $\bar{\mathcal{X}}_f := \Delta \oplus \{\bar{x}_{\infty}\}$.

Finally, for condition~\eqref{sc}, we require:
\begin{equation}
\begin{split}
(A\bar{x}_k + B\kappa(\bar{x}_k) - \bar{x}_{\infty})^TP(A\bar{x}_k + B\kappa(\bar{x}_k) - \bar{x}_{\infty})& \\ + \delta \bar{x}_k^TQ\delta \bar{x}_k + (\kappa(\bar{x}_k)-\bar{u}_{\infty})^TR(\kappa(\bar{x}_k)-\bar{u}_{\infty})&  \\
\leq \delta \bar{x}_k^TP\delta \bar{x}_k&.
\end{split}
\end{equation}
Substituting the terminal control law (\ref{termcont}), we require:
\begin{equation}
\begin{split}
\delta \bar{x}_k^T A_{K_f}&^T P A_{K_f} \delta \bar{x}_k + \delta \bar{x}_k^TQ\delta \bar{x}_k \\
+&\delta \bar{x}_k^TA_{K_f}^T R A_{K_f} \delta \bar{x}_k \leq \delta \bar{x}_k^T P \delta \bar{x}_k.
\end{split}
\end{equation}
Since the chosen $P$ and $K_f$ are optimal solutions of the Riccati equation, this expression holds with equality.

By the results in \cite{RawlingsMayne2017}, the ROMPC problem (\ref{optimalcontrol}) will be recursively feasible and exponentially stable, i.e., $\delta \bar{x} \rightarrow 0$ and $\delta \bar{u} \rightarrow 0$.

\section{Ensuring Robust Constraint Satisfaction} \label{sec:errorbounds}
Guaranteeing that the constraints (\ref{const}) on the \textit{full order system} remain satisfied when using reduced order MPC is challenging due to disturbances, model reduction error, and estimation error. In this section we discuss how to compute bounds for these errors and ensure robust constraint satisfaction by computing tightened constraint sets $\bar{\mathcal{Z}}$ and $\bar{\mathcal{U}}$ for the ROMPC problem (\ref{optimalcontrol}). 

Consider writing the original constraints (\ref{const}) as
\begin{equation*}
\begin{split}
H^fx^f_k \in \mathcal{Z} &\iff \hat{e}_k + Hd_k + H\bar{x}_k \in \mathcal{Z}, \\
u_k \in \mathcal{U} &\iff \bar{u}_k + Kd_k \in \mathcal{U},
\end{split}
\end{equation*}
where $\hat{e}_k = H^fx^f_k -H\hat{x}_k$ will be referred to as the estimation error (for the performance variable) and $d_k = \hat{x}_k - \bar{x}_k$ will be referred to as the control error.
From these equations it is apparent that by finding bounds on the estimation and control errors, the performance and control constraints can be appropriately tightened to give constraints on $\bar{x}_k$ and $\bar{u}_k$. In this work we extend the approach used in \cite{KoegelFindeisen2015b} to compute tightened constraint sets, $\bar{\mathcal{Z}}$ and $\bar{\mathcal{U}}$, that are less conservative.

\subsection{Control Error Bound} \label{contbound}
For the control constraints, consider a tightened constraint set $\bar{\mathcal{U}} := \mathcal{U} \ominus K\mathcal{D}$, where  $\mathcal{D}$ is a set bounding the control error $d_k$ and $\ominus$ denotes the Pontryagin difference.
The dynamics of the control error, using the control law (\ref{controller}), are
\begin{equation*}
d_{k+1} = (A+BK)d_k + LC^fx^f_k - LC\hat{x}_k + Lv_k.
\end{equation*}

In \cite{KoegelFindeisen2015b}, the approach is to compute a bound $\mathcal{G}$ on an auxiliary disturbance $g_k = LC^fx^f_k - LC\hat{x}_k + Lv_k$ by solving linear programs\footnote{Note that in \cite{KoegelFindeisen2015b} a fixed interval least square estimator is used, so $g_k$ takes on a different form.}. Then, a bounding set $\mathcal{D}$ is computed that satisfies the condition $(A+BK)\mathcal{D} + \mathcal{G} \subseteq \mathcal{D}$ (see \cite{RakovicKerriganEtAl2005}).

Finding the bound $\mathcal{G}$ and computing $\mathcal{D}$ in a \textit{sequential} fashion introduce some conservatism. We instead compute the bounding set $\mathcal{D}$ \textit{directly} by leveraging the fact that the dynamics of $g_k$ are known through the full order system dynamics (\ref{fom}) and the observer dynamics (\ref{estimator}). This reduces conservatism by removing the worst-case consideration of the auxiliary disturbance $g_k$.
In particular, our approach is to approximate $\mathcal{D}$ by a bounded convex polytope where each side is defined by solving a linear program of the form
\begin{subequations} \label{computeD} 
\begin{align} 
&\underset{d_i, \bar{x}_j, \hat{x}_j, x^f_j, u_i, v_i, w^f_i}{\text{maximize}} \:\:  \theta^{\mathcal{D}^T}_l (\hat{x}_k - \bar{x}_k) \tag{\ref{computeD}} \\
&\text{subject to} \:\: \nonumber \\
&x^f_{i+1} = A^fx^f_i + B^fu_i + w^f_i, \label{cD:1} \\
&\hat{x}_{i+1} = (A-LC)\hat{x}_i + Bu_i + LC^fx^f_i + Lv_i, \label{cD:2}\\
&\bar{x}_{i+1} = A\bar{x}_i + B(u_i - K(\hat{x}_i - \bar{x}_i)) \label{cD:3}\\
&d_i = \hat{x}_i - \bar{x}_i, \label{cD:4}\\
&H^fx^f_i \in \mathcal{Z}, \:\:
u_i \in \mathcal{U}, \:\:
v_i \in \mathcal{V}, \:\:
w^f_i \in \mathcal{W}, \label{cD:5}\\
&H\bar{x}_i \in \mathcal{Z}, \:\:
d_i \in \mathcal{D}_0, \label{cD:6} 
\end{align}
\end{subequations}
where the decision variables are $d_i$, $\bar{x}_j$, $\hat{x}_j$, $x^f_j$, $u_i$, $v_i$, $w^f_i$, $i \in [k-\tau,\dots,k-1]$, and $j \in [k-\tau,\dots,k]$. The time horizon $\tau$ and the set $\mathcal{D}_0$ are user-defined parameters (discussed later). The vectors $\theta^{\mathcal{D}}_l$ define the direction normal to the $l^{\text{th}}$ face of the bounding polytope, which are also design parameters (e.g., the standard basis vectors).

The constraints $d_i \in \mathcal{D}_0$ and $H\bar{x}_i \in \mathcal{Z}$ are necessary for the linear program to be bounded. Note that $H\bar{x}_i \in \mathcal{Z}$ is simply a conservative estimate of the constraint imposed in the ROMPC scheme that $H\bar{x} \in \bar{\mathcal{Z}}$ (since $\bar{\mathcal{Z}}$ is a tightened set of $\mathcal{Z}$ it follows that $\bar{\mathcal{Z}} \subseteq \mathcal{Z}$).
The set $\mathcal{D}$, parameterized by $(\tau,\mathcal{D}_0)$, is then given by the polytope  $\mathcal{D}(\tau, \mathcal{D}_0)$ defined as:
\begin{equation} \label{setD}
\mathcal{D}(\tau, \mathcal{D}_0) := \{d \: | \: \Theta^{\mathcal{D}} d \leq \gamma^{\mathcal{D}} \},
\end{equation}
where the $l^{\text{th}}$ row of $\Theta^{\mathcal{D}}$ is given by $\theta^{\mathcal{D}^T}_l$, and $\gamma^{\mathcal{D}}_l$ is the optimal value of the linear program (\ref{computeD}). 

To illustrate the advantage of computing the bound using the linear program (\ref{computeD}) we consider the system described in Section \ref{sec:numericalex}. First, the set $\mathcal{D}$ is computed as discussed in \cite{KoegelFindeisen2015b} and is displayed in red in Figure \ref{setcompare}. Second, the set $\mathcal{D}$ is found using the proposed method and is displayed in green in Figure \ref{setcompare}. It can be seen that the bound on $d_k$ is significantly tighter for this system using our modified method.

\subsection{Total Error Bound} \label{totbound}
For the performance variable constraints, we consider a tightened constraint set $\bar{\mathcal{Z}} := \mathcal{Z} \ominus \mathcal{E}$, where  $\mathcal{E}$ is a set bounding the total error $e_k = \hat{e}_k + Hd_k$. The approach in \cite{KoegelFindeisen2015b} is to compute a bound $\hat{\mathcal{E}}$ on $\hat{e}_k$ by solving linear programs, and then to compute $\mathcal{E}$ using a Minkowski sum $\mathcal{E} = \hat{\mathcal{E}} \oplus H\mathcal{D}$. Once again, \textit{sequentially} computing $\hat{\mathcal{E}}$ and $\mathcal{D}$ introduces conservatism because it ignores the fact that the dynamics of $\hat{e}$ and $d$ are coupled. Therefore we \textit{directly} compute $\mathcal{E}$ as a bounded convex polytope, where each side is defined by solving a linear program of the form
\begin{align} \label{computeE}
&\underset{d_i, \bar{x}_j, \hat{x}_j, x^f_j, u_i, v_i, w^f_i}{\text{maximize}} \:\: \theta^{\mathcal{E}^T}_l(H^f x^f_k - H \bar{x}_k) \\
&\text{subject to} \:\:  \eqref{cD:1}\text{-}\eqref{cD:6}, \nonumber
\end{align}
where the decision variables are $d_i$, $\bar{x}_j$, $\hat{x}_j$, $x^f_j$, $u_i$, $v_i$, $w^f_i$, $i \in [k-\tau,\dots,k-1]$, and $j \in [k-\tau,\dots,k]$. Note that $e_k = H^fx^f_k -H\bar{x}_k$, and $\tau$, $\mathcal{D}_0$ are the same parameters used when solving the linear program (\ref{computeD}) for $\mathcal{D}$. The vectors $\theta^{\mathcal{E}}_l$ define the direction normal to the $l^{\text{th}}$ face of the bounding polytope. The set $\mathcal{E}$ is then given by the polytope  $\mathcal{E}(\tau, \mathcal{D}_0)$ defined as:
\begin{equation} \label{setE}
\mathcal{E}(\tau, \mathcal{D}_0) := \{e \: | \: \Theta^{\mathcal{E}} e \leq \gamma^{\mathcal{E}} \},
\end{equation}
where the $l^{\text{th}}$ row of $\Theta^{\mathcal{E}}$ is given by $\theta^{\mathcal{E}^T}_l$ and $\gamma^{\mathcal{E}}_l$ is the optimal value of the linear program (\ref{computeE}). 

To illustrate the advantage of computing the total error bound $\mathcal{E}$ \textit{directly} we once again consider the system described in Section \ref{sec:numericalex}. For comparison, first $\mathcal{E} = \hat{\mathcal{E}} \oplus H\mathcal{D}$ is computed using the approach in \cite{KoegelFindeisen2015b} and is plotted in red in Figure \ref{setcompare}. Second, $\mathcal{E} = \hat{\mathcal{E}} \oplus H\mathcal{D}$ is computed where $\mathcal{D}$ is defined using our proposed method (i.e., from (\ref{computeD})), and is plotted in blue. Finally, the total error bound $\mathcal{E}$ is computed \textit{directly} using linear programs of the form (\ref{computeE}) and is plotted in green. It is apparent that direct computation of $\mathcal{E}$ provides a tighter polytopic approximation to the bound on the error $e_k$ than the other approaches. Note that computation of these linear programs (\ref{computeD}) and (\ref{computeE}) relies on the full order dynamics, which could potentially lead to large linear programs. However, these computations only need to be done once, and are done offline. 

\begin{figure}[t]
    \centering
    \begin{subfigure}[t]{.32\textwidth}
        \includegraphics[width=\textwidth]{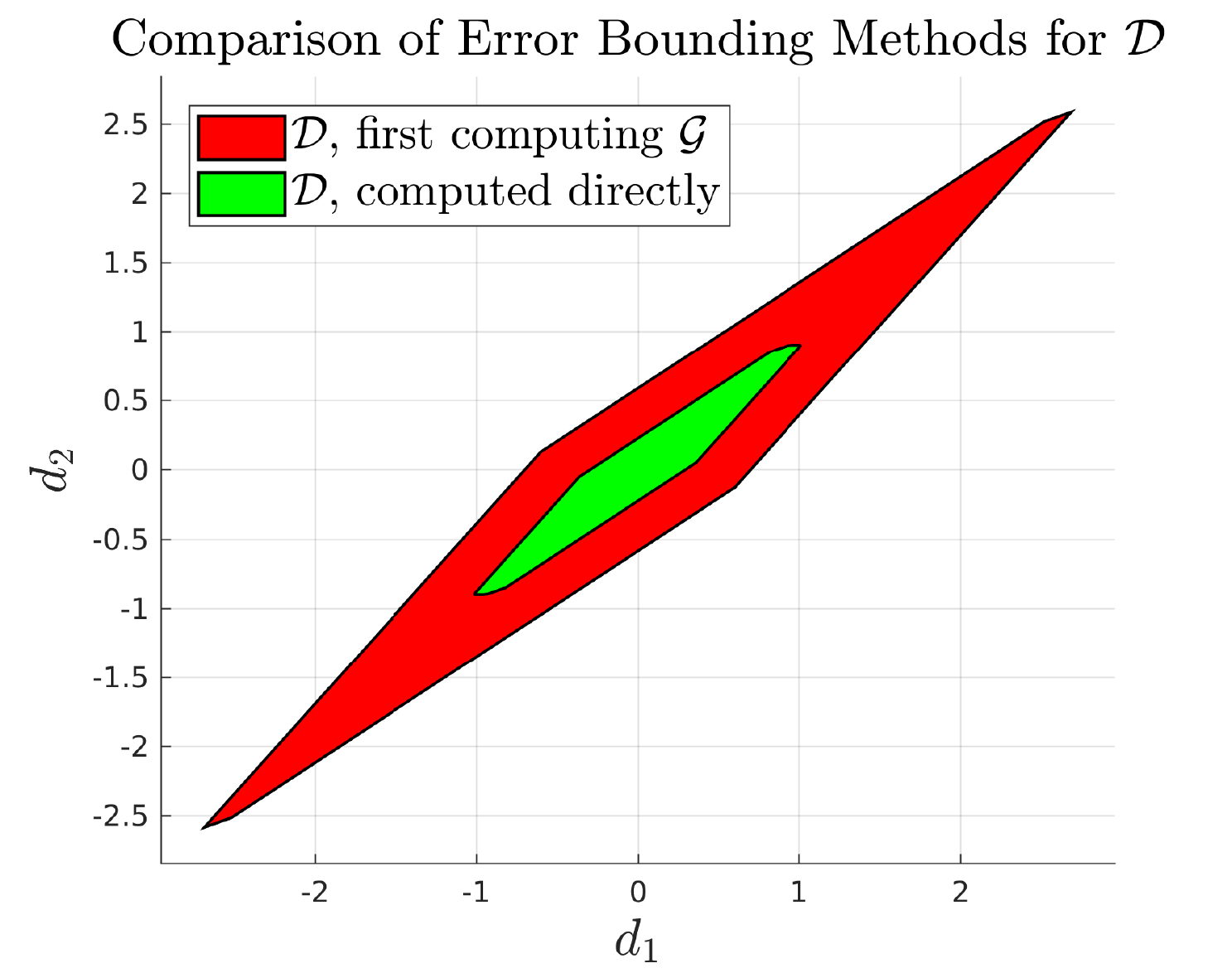}
    \end{subfigure}
    \begin{subfigure}[t]{.32\textwidth}
        \includegraphics[width=\textwidth]{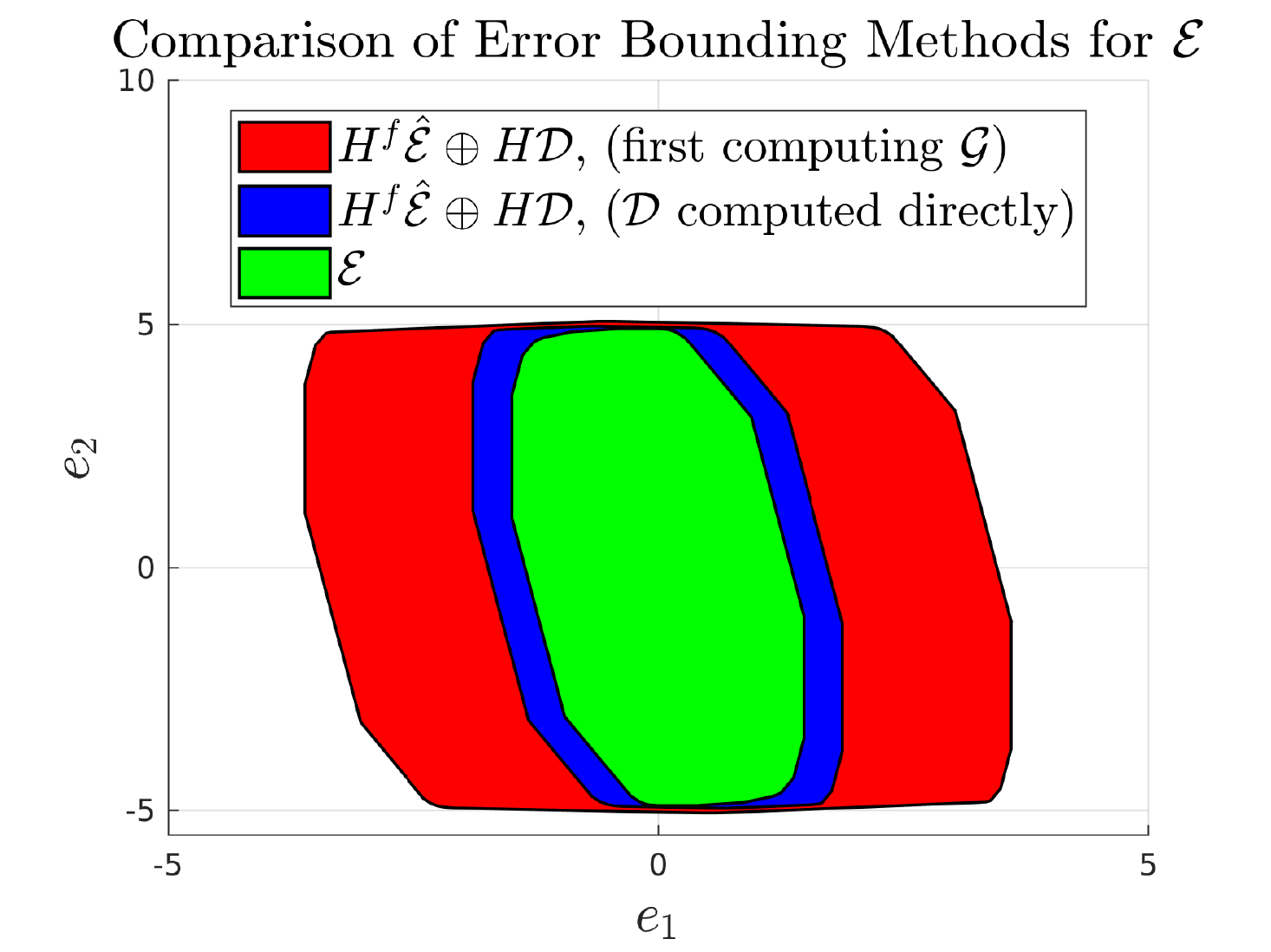}
    \end{subfigure}
    \caption{(Above) The error bounds $\mathcal{D}$ using a previously proposed method (red) and modified method (green). (Below) The error bounds used to tighten $\mathcal{Z}$ using a previously proposed method (red), combined method (blue), and modified method (green).}
\label{setcompare}
\vspace{-.20in}
\end{figure}

\subsection{Robust Constraint Satisfaction}
Under mild assumptions we now show that with the tightened constraints $\bar{\mathcal{Z}}$ and $\bar{\mathcal{U}}$ resulting from the pre-computed bounds $\mathcal{D} = \mathcal{D}(\tau, \mathcal{D}_0)$ and $\mathcal{E} = \mathcal{E}(\tau, \mathcal{D}_0)$, the combined control system consisting of the state estimator (\ref{estimator}), the control law (\ref{controller}), and the ROMPC problem (\ref{optimalcontrol}) guarantees robust constraint satisfaction for the \textit{full order system}.
\begin{ass} \label{consistency}
At some arbitrary time $k=k_d$, the full order system has followed an admissible trajectory $\{x^f_i\}^{k_d}_{i=k_d-\tau}$ over the past $\tau$ time steps, satisfying the dynamics (\ref{fom}) and constraints (\ref{const}) under admissible disturbances (\ref{noise}).
\end{ass}
\begin{ass} \label{romconsistency}
With the control $\{u_i\}^{k_d-1}_{i=k_d-\tau}$ and measurement $\{y_i\}^{k_d}_{i=k_d-\tau}$ sequences corresponding to the admissible trajectory in Assumption \ref{consistency}, the state estimator sequence $\{\hat{x}_i\}^{k_d}_{i=k_d-\tau}$ and the nominal reduced order system trajectory $\{\bar{x}_i\}^{k_d}_{i=k_d-\tau}$ under the control law $\bar{u}_i = u_i - K(\hat{x}_i-\bar{x}_i)$ satisfy $H\bar{x}_i \in \mathcal{Z}$ and $\hat{x}_i - \bar{x}_i \in \mathcal{D}_0$ for $i \in [k_d-\tau,\dots,k_d]$.
\end{ass}
\begin{ass} \label{feasibleMPC}
At time $k = k_d$ the ROMPC problem (\ref{optimalcontrol}), with the initial state $\bar{x}_{k_d}$ defined by the sequence $\{\bar{x}_i\}^{k_d}_{i=k_d-\tau}$ from Assumption \ref{romconsistency}, has a feasible solution.
\end{ass}

Note that Assumption  \ref{consistency} is often valid in practice, for example if the system starts at any feasible steady state. 
Additionally, because the estimator and the nominal reduced order system are not physical systems, Assumption \ref{romconsistency} can be satisfied by artificially choosing any $\hat{x}_{k_d-\tau}$ and $\bar{x}_{k_d-\tau}$ that yield sequences $\{\hat{x}_i\}^{k_d}_{i=k_d-\tau}$, $\{\bar{x}_i\}^{k_d}_{i=k_d-\tau}$ satisfying the assumption conditions.
\begin{thm}[Robust Constraint Satisfaction] \label{Robusttheorem}
For an arbitrary time $k=k_d$ let Assumptions  \ref{consistency}, \ref{romconsistency}, and \ref{feasibleMPC} hold. If for $k \geq k_d$ the control law (\ref{controller}) is used to control the full order system where $\bar{\mathcal{U}} = \mathcal{U}\ominus K\mathcal{D}$ and $\bar{\mathcal{Z}} = \mathcal{Z}\ominus \mathcal{E}$ in the ROMPC problem (\ref{optimalcontrol}), and if $\mathcal{D}$ and $\mathcal{E}$ are defined by solving the linear programs (\ref{computeD}) and (\ref{computeE}) and $\mathcal{D} \subseteq \mathcal{D}_0$, then the constraints (\ref{const}) will be robustly satisfied for all $k \geq k_d$.
\end{thm}

\begin{proof}
From Assumptions \ref{romconsistency} and \ref{feasibleMPC}, and by virtue of the recursive feasibility of the ROMPC problem (\ref{optimalcontrol}), the nominal reduced order state satisfies $H\bar{x}_k \in \bar{\mathcal{Z}}$ and the nominal control satisfies $\bar{u}_k \in \bar{\mathcal{U}}$ for all $k \geq k_d$. Therefore, since $\bar{\mathcal{U}} = \mathcal{U}\ominus K \mathcal{D}$ and $\bar{\mathcal{Z}} = \mathcal{Z}\ominus \mathcal{E}$, to show $u_k \in \mathcal{U}$ and $z_k \in \mathcal{Z}$ we must show that $d_k \in \mathcal{D}$ and $e_k \in \mathcal{E}$. We prove that these inclusions hold for all $k \geq k_d$ via induction.

From the construction of the linear programs (\ref{computeD}) and (\ref{computeE}), the inclusions $d_k \in \mathcal{D}(\tau, \mathcal{D}_0)$ and $e_k \in \mathcal{E}(\tau, \mathcal{D}_0)$ are valid if the following conditions are satisfied for $i \in [k-\tau,\dots,k-1]$: (i) $H^fx^f_i \in \mathcal{Z}$, (ii) $u_i \in \mathcal{U}$, (iii) $H\bar{x}_i \in \mathcal{Z}$, and (iv) $d_i \in \mathcal{D}_0$. At time $k=k_d$ it can trivially be seen that conditions (i)-(iv) are satified from Assumptions \ref{consistency} and \ref{romconsistency}, and therefore $d_{k_d} \in \mathcal{D}$ and $e_{k_d} \in \mathcal{E}$. 

Now assume the conditions (i)-(iv) hold at some arbitrary time $k\geq k_d$ such that $d_k \in \mathcal{D}$ and $e_{k} \in \mathcal{E}$. Then since $\bar{\mathcal{Z}} \subseteq \mathcal{Z}$ and $\mathcal{D} \subseteq \mathcal{D}_0$, and since $H\bar{x}_k \in \bar{\mathcal{Z}}$ and $\bar{u}_k \in \bar{\mathcal{U}}$ by design of the ROMPC problem solved at time $k$ we have: (a) $H^fx^f_k \in \mathcal{Z}$, (b) $u_k \in \mathcal{U}$, (c) $H\bar{x}_k \in \mathcal{Z}$, and (d) $d_k \in \mathcal{D}_0$. Therefore conditions (i)-(iv) are satisfied at time $k + 1$ such that $d_{k+1} \in \mathcal{D}$ and $e_{k+1} \in \mathcal{E}$, which completes the induction.
\end{proof}

\subsection{Practical Considerations}
It is important to note that several conditions on the sets $\mathcal{D}$ and $\mathcal{E}$ must be met in order to apply the proposed control scheme in practice, namely: $\mathcal{D} \subseteq \mathcal{D}_0$ (Theorems \ref{Robusttheorem}), $\bar{\mathcal{Z}} = \mathcal{Z} \ominus \mathcal{E} \not= \emptyset$, and $\bar{\mathcal{U}} = \mathcal{U} \ominus K\mathcal{D} \not= \emptyset$. A practitioner can compute a satisfactory $\mathcal{D}$, $\mathcal{E}$ using the following algorithm:
\begin{algorithm}
\caption{Compute $\mathcal{D}$ and $\mathcal{E}$}\label{alg}
\begin{algorithmic}[1]
\Procedure{ComputeSets}{$\mathcal{Z}$, $\mathcal{U}$, $K$, $\tau$}
\State $\mathcal{D}_0 \leftarrow \{d \: | \: Kd \in \mathcal{U} \}$
\State $\mathcal{D} \leftarrow \mathcal{D}(\tau, \mathcal{D}_0)$, using (\ref{computeD})
\State $\mathcal{E} \leftarrow \mathcal{E}(\tau, \mathcal{D}_0)$, using (\ref{computeE})
\State $\bar{\mathcal{Z}} \leftarrow \mathcal{Z} \ominus \mathcal{E}$
\If {$\mathcal{D} \not\subseteq \mathcal{D}_0$ or $\bar{\mathcal{Z}} = \emptyset$}{ return Failure}
\Else { return $\mathcal{D}_0, \mathcal{D}, \mathcal{E}$}
\EndIf
\EndProcedure
\end{algorithmic}
\end{algorithm}

In this algorithm, $\mathcal{D}_0$ is defined as the largest possible set such that $\bar{\mathcal{U}} = \mathcal{U} \ominus K\mathcal{D}$ is guaranteed to be non-empty if $\mathcal{D} \subseteq \mathcal{D}_0$. Note that when $\mathcal{U}$ is polytopic, it can be defined as $\mathcal{U} = \{u \: | \: A_Uu \leq b_u \}$ and $\mathcal{D}_0$ in Algorithm \ref{alg} can be defined as $\mathcal{D}_0 = \{d \: | \: A_UKd \leq b_u \}$.

If for some $\tau$ Algorithm \ref{alg} is successful, a valid $\mathcal{D}$ and $\mathcal{E}$ have been found. The value of $\tau$ can then be continually increased and Algorithm \ref{alg} can be continually run to attempt to find smaller bounds. This is desired because smaller $\mathcal{D}$, $\mathcal{E}$ lead to less conservative constraint tightening. If Algorithm \ref{alg} returns ``Failure'', the value of $\tau$ should also be increased. However, if Algorithm \ref{alg} continues to return ``Failure'' as $\tau \rightarrow \infty$ then additional diagnostics are required. These diagnostics will attempt to determine whether the failure to guarantee robust constraint satisfaction is due to the process and measurement disturbances (\ref{noise}), or due to the design of the reduced order model (\ref{rom}). First, the practitioner should set the disturbances (\ref{noise}) to zero and again apply Algorithm \ref{alg}. If Algorithm \ref{alg} is then successful for some value $\tau$, this suggests that the disturbances (\ref{noise}) need to be reduced, either by using a higher fidelity full order model or by adding additional/more precise sensors. However if no value of $\tau$ can be found that makes Algorithm \ref{alg} successful, the practitioner should consider redesigning the reduced order model.

\section{Setpoint Tracking}\label{sec:tracking}
Having addressed ROMPC stability (Section \ref{sec:stability}) and robust constraint satisfaction (Section \ref{sec:errorbounds}), we now present results on the setpoint tracking performance. Specifically, we will show that under the bounded disturbances (\ref{noise}) the tracking variables $z^r_k$ converge to the set $\{r\} \oplus \mathcal{R}$ where $\mathcal{R}$ is a bounded convex polytope. Further, we will show that with no disturbances the ROMPC scheme converges to offset-free (i.e., zero error) tracking.

We begin in Section \ref{trackingcond} by discussing conditions for offset-free setpoint tracking under nominal (disturbance free) conditions. In this section we also describe the computation of the target states $\bar{x}_{\infty}$, $\bar{u}_{\infty}$ used in the ROMPC objective function (\ref{cost}). Then, in Section \ref{setpointconverge} we discuss the computation of the set $\mathcal{R}$ and convergence of $z^r_k$ to $\{r\} \oplus \mathcal{R}$ under bounded disturbances.

\subsection{Offset-Free Setpoint Tracking} \label{trackingcond}
In this section we discuss how to compute the ROMPC target values $\bar{x}_{\infty}$ and $\bar{u}_{\infty}$ that enable offset-free tracking at steady state for the nominal system (i.e., when no disturbances are present).

First, the full order target steady state $x^f_{\infty}$ and control $u_{\infty}$ are computed such that $z^r_{\infty} = r$ by finding the solution to the linear system (\ref{steadystate}) (assuming the desired setpoint $r$ corresponds to an admissible steady state $H^f x^f_{\infty} \in \mathcal{Z}$, $u_{\infty} \in \mathcal{U}$). For the system to reach the steady state defined by $x^f_{\infty}$, $u_{\infty}$, the observer and controller are also required to be at steady state. In the absence of disturbances, the steady-state output of (\ref{fom}) is $y_{\infty} = C^fx^f_{\infty}$. Therefore, from the observer dynamics (\ref{estimator}), the steady-state observer estimate is given by:
\begin{equation} \label{eq:estSS}
\hat{x}_{\infty} = D(Bu_{\infty} + LC^fx^f_{\infty}),
\end{equation}
where $D = (I-(A-LC))^{-1}$.

Finally, by requiring the controller (\ref{controller}) to also be at steady state, the ROMPC target states $\bar{x}_{\infty}$ and $\bar{u}_{\infty}$ that enable offset-free setpoint tracking under nominal conditions can be found by solving the system
\begin{equation} \label{nomsteadystate}
S_c \begin{bmatrix}
\bar{x}_{\infty} \\ \bar{u}_{\infty}
\end{bmatrix} = 
\begin{bmatrix}
0 \\ K\hat{x}_{\infty} - u_{\infty}
\end{bmatrix}, \quad S_c = \begin{bmatrix}
A-I & B\\ K & -I
\end{bmatrix},
\end{equation}
where it is assumed that the square matrix $S_c$ is full rank and that the ROMPC target states are feasible with respect to the tightened constraints sets: $H\bar{x}_{\infty} \in \bar{\mathcal{Z}}$, $\bar{u}_{\infty} \in \bar{\mathcal{U}}$. Note that when the steady-state equations (\ref{steadystate}), (\ref{eq:estSS}), and (\ref{nomsteadystate}) are solved, a \textit{unique} solution ($x^f_{\infty}$, $u_{\infty}$, $\hat{x}_{\infty}$, $\bar{x}_{\infty}$, $\bar{u}_{\infty}$) exists for each setpoint $r$ because the square matrix 
\begin{equation}
F=
\begin{bmatrix}
A^f - I & B^f & 0 & 0 & 0 \\
TH^f & 0 & 0 & 0 & 0 \\
DLC^f & DB & -I & 0 & 0 \\
0 & 0 & 0 & A-I & B \\
0 & I & -K & K & -I \\
\end{bmatrix}
\end{equation}
is full rank. This follows from the block lower diagonal structure of $F$ and the existing requirements that the square matrices $S_f$ and $S_c$ are full rank. 

With a unique steady state and with the ROMPC convergence results in Section \ref{sec:stability}, if the closed-loop system reaches a steady state, then it must be the unique steady state that leads to offset-free tracking. To determine if the system will reach a steady state, consider the closed-loop dynamics of the errors $\delta x^f_k = x^f_k - x^f_{\infty}$ and $\delta \hat{x}_k = \hat{x}_k - \hat{x}_{\infty}$ under the controller (\ref{controller}) with $\delta \bar{x} = 0$ and $\delta \bar{u} = 0$ (the general case with $\delta \bar{x}, \delta \bar{u} \neq 0$ has no effect on the following conclusions as these errors would only contribute to an asymptotically decaying term added to the following equation):
\begin{equation}
\begin{bmatrix}
\delta x^f_{k+1} \\ \delta \hat{x}_{k+1}
\end{bmatrix} = S_{ss}\begin{bmatrix}
\delta x^f_k \\ \delta \hat{x}_k
\end{bmatrix},
\end{equation}
where the matrix $S_{ss}$ is given by
\begin{equation}
S_{ss}=\begin{bmatrix}
A^f & B^fK \\
LC^f & A-LC+BK \\
\end{bmatrix}.
\end{equation}
From these dynamics, a sufficient condition for the closed-loop system to reach steady state is if the matrix $S_{ss}$ is Schur stable. 
This condition provides a straightforward method for a practitioner to check if convergence to offset-free tracking is guaranteed in the disturbance free case. However, synthesizing the reduced order model, controller, and observer that ensure $S_{ss}$ is Schur stable is not straightforward and is a planned area of future work. 

To address the cases when the closed-loop system does not reach a steady state, we now present a more general analysis of the setpoint tracking performance.

\subsection{Tracking Variable Convergence} \label{setpointconverge}
Previously we introduced conditions that guarantee offset-free setpoint tracking under nominal (disturbance free) scenarios. We now use these results to discuss convergence of the tracking variables when the closed-loop system does not reach steady state (e.g., due to \textit{bounded} disturbances). Specifically, we can show that under the proposed control scheme the tracking variables converge to a set containing the desired setpoint, $\{r\} \oplus \mathcal{R}$. This is accomplished by noting that as the ROMPC scheme converges the control law (\ref{controller}) will compensate for the nominal model reduction error. 

To characterize the set $\mathcal{R}$, first, let $\epsilon_x, \epsilon_u > 0$ be user-defined convergence thresholds for the ROMPC problem. From Section~\ref{sec:stability}, since $(\delta \bar{x},\delta \bar{u}) = (0,0)$ is exponentially stable, the conditions $||\delta \bar{x}||_\infty \leq \epsilon_x$ and $||\delta \bar{u}||_\infty \leq \epsilon_u$ are attainable in finite time. 
Now consider the linear program:
\begin{equation} \label{computeR}
\begin{split}
&\underset{x^f_j, \hat{x}_j, u_i, v_i, w^f_i, r, (x^f, u, \hat{x}, \bar{x}, \bar{u})_{\infty}}{\text{maximize}} \:\: \theta^{\mathcal{R}^T}_l(TH^fx^f_k - r) \\
&\text{subject to} \:\: \\
&x^f_{i+1} = A^fx^f_i + B^fu_i + w^f_i, \\
&\hat{x}_{i+1} = (A-LC)\hat{x}_i + Bu_i + LC^fx^f_i + Lv_i, \\
&u_i = \bar{u}_i + K(\hat{x}_i - \bar{x}_i), \\
&||\bar{u}_i - \bar{u}_{\infty}||_\infty \leq \epsilon_u, \:\: ||\bar{x}_i - \bar{x}_{\infty}||_\infty \leq \epsilon_x, \\
&\hat{x}_j - \bar{x}_{\infty} \in \mathcal{D}, \:\: H^fx^f_j \in \mathcal{Z}, \:\:
v_i \in \mathcal{V}, \:\:
w^f_i \in \mathcal{W},\\
&\begin{bmatrix}
x^f_{\infty} \\ u_{\infty}
\end{bmatrix} = 
S_f^{-1}\begin{bmatrix}
0 \\ r
\end{bmatrix}, \\
&\hat{x}_{\infty} = (I-(A-LC))^{-1}\Big(Bu_{\infty} + LC^fx^f_{\infty}\Big), \\
&\begin{bmatrix}
\bar{x}_{\infty} \\ \bar{u}_{\infty}
\end{bmatrix} = 
S_c^{-1}\begin{bmatrix}
0 \\ K\hat{x}_{\infty} - u_{\infty}
\end{bmatrix}, \\
&H^fx^f_{\infty} \in \mathcal{Z}, \:\: u_{\infty} \in \mathcal{U}, \:\: H\bar{x}_{\infty} \in \bar{\mathcal{Z}}, \:\: \bar{u}_{\infty} \in \bar{\mathcal{U}},\\
\end{split}
\end{equation}
where $i = k-\tau_{ss},\dots,k-1$, $j = k-\tau_{ss},\dots,k$ and the decision variables are $x^f_j$, $\hat{x}_j$, $u_i$, $v_i$, $w^f_i$, $r$, $\bar{x}_i$, $\bar{u}_i$, and the steady-state variables $x^f_{\infty}$, $u_{\infty}$, $\hat{x}_{\infty}$, $\bar{x}_{\infty}$, $\bar{u}_{\infty}$. The parameters of the problem include the sets $\mathcal{D}$, $\bar{\mathcal{Z}}$, and $\bar{\mathcal{U}}$ (computed in Section \ref{sec:errorbounds}), the time horizon, $\tau_{ss}$ (which can be different from $\tau$, used in Section \ref{sec:errorbounds}), and the thresholds $\epsilon_x$, $\epsilon_u$. The vectors $\theta^{\mathcal{R}}_l$ define the direction normal to the $l^{\text{th}}$ face of the bounding polytope. The set $\mathcal{R}$ is then given by the polytope $\mathcal{R}(\tau_{ss}, \mathcal{D},\bar{\mathcal{Z}},\bar{\mathcal{U}},\epsilon_x,\epsilon_u)$ defined as:
\begin{equation} \label{setR}
\mathcal{R}(\tau_{ss}, \mathcal{D},\bar{\mathcal{Z}},\bar{\mathcal{U}},\epsilon_x,\epsilon_u) := \{e_r \: | \: \Theta^{\mathcal{R}} e_r \leq \gamma^{\mathcal{R}} \},
\end{equation}
where $e_r = z^r - r$ and the $l^{\text{th}}$ row of $\Theta^{\mathcal{R}}$ is given by $\theta^{\mathcal{R}^T}_l$ and $\gamma^{\mathcal{R}}_l$ is the optimal value of the linear program (\ref{computeR}). Note that the desired setpoint $r$ is a decision variable in this problem, which makes the computed set $\mathcal{R}$ valid for \textit{all} feasible setpoints. 

By employing the results on robust constraint satisfaction from Section \ref{sec:errorbounds} and the results on ROMPC stability from Section \ref{sec:stability}, we can now present our main result.
\begin{thm}[Setpoint Tracking Under Bounded Disturbances] \label{Rtheorem}
Let the ROMPC problem be given by (\ref{optimalcontrol}), and the conditions from Theorem \ref{Robusttheorem} hold. Let the set $\mathcal{R} = \mathcal{R}(\tau_{ss},\mathcal{D},\bar{\mathcal{Z}},\bar{\mathcal{U}},\epsilon_x,\epsilon_u)$ be defined by solving (\ref{computeR}) for a given $\tau_{ss}$, $\mathcal{D}$, $\bar{\mathcal{Z}}$, $\bar{\mathcal{U}}$, $\epsilon_x$, $\epsilon_u$.
If the tracking targets $x^f_{\infty}$, $u_{\infty}$ are solutions to (\ref{steadystate}) and the ROMPC targets $\bar{x}_{\infty}$, $\bar{u}_{\infty}$ are solutions to (\ref{nomsteadystate}) and satisfy the constraints $H\bar{x}_{\infty} \in \bar{\mathcal{Z}}$, $\bar{u}_{\infty} \in \bar{\mathcal{U}}$, then there exists a finite time $k_r$ such that the tracking variable $z^r_k$ will lie within the set $\{r\} \oplus \mathcal{R}$ for all time $k \geq k_r + \tau_{ss}$.
\end{thm}
\begin{proof}
From the linear program (\ref{computeR}), the inclusion $z^r_k \in \{r\} \oplus \mathcal{R}$ is valid if the following conditions are satisfied: (i) $||\delta \bar{x}_i||_\infty \leq \epsilon_x$ and $||\delta \bar{u}_i||_\infty \leq \epsilon_u$ for $i \in [k-\tau_{ss},\dots,k-1]$, (ii) $\hat{x}_j - \bar{x}_{\infty} \in \mathcal{D}$, and (iii) $H^fx^f_i \in \mathcal{Z}$ for $j \in [k-\tau_{ss},\dots,k]$. By the exponential convergence of ROMPC (illustrated in Section~\ref{sec:stability}), there exists a finite time $k_r$ such that condition (i) is satisfied for all $k \geq k_r + \tau_{ss}$. Theorem \ref{Robusttheorem} guarantees that condition (ii) is satisfied for all $k \geq k_r + \tau_{ss}$ when $k_r \geq k_d$. Finally, condition (iii) is also satisfied by Theorem \ref{Robusttheorem} for all $k \geq k_d$.
\end{proof}

\begin{crl}[Offset-free Setpoint Tracking]
If no disturbances act on the system and the matrix $S_{ss}$ is Schur stable, then the tracking variables will converge to the setpoint with zero offset (i.e., $z^r_k \rightarrow r$ as $k\rightarrow \infty$).
\end{crl}

\section{Numerical Experiments} \label{sec:examples}
Next we discuss two example applications of our approach. The first is a synthetic system taken from the reduced order modeling literature, while the second addresses a more complex task of controlling a flexible beam.

\subsection{Synthetic System} \label{sec:numericalex}
We first consider a numerical example adapted from \cite{HovlandLovaasEtAl2008}. The full order system has dimension $n^f=6$ and is given by the matrices $A^f$, $B^f$ in \cite{HovlandLovaasEtAl2008}, and
\begin{equation*}
C^f = \begin{bmatrix}
1.29 & 0.24 & 0_{1\times 4}
\end{bmatrix}, \quad 
H^f = \begin{bmatrix}
I_{2\times 2} & 0_{2\times 4}
\end{bmatrix}.
\end{equation*}

The performance and control constraints are $\mathcal{Z} = \{z \: | \:\:\: ||z||_{\infty} \leq 10 \}$ and $\mathcal{U} = \{u \: | \:\:\: ||u||_{\infty} \leq 
2 \}$. The process and measurement noise are bounded by $\mathcal{W} = \{w^f \: | \:\:\: ||w^f||_{\infty} \leq 0.05 \}$ and $\mathcal{V} = \{v \: | \:\:\: ||v^f||_{\infty} \leq 0.05 \}$. The reduced order model has dimension $n=2$ and was computed using balanced truncation.

In this example, both the proposed approach and a na\"ive approach  are implemented to demonstrate that setpoint tracking is non-trivial when using reduced order models for control. The na\"ive approach is to simply choose the desired reduced order target states $\bar{x}_{\infty}$, $\bar{u}_{\infty}$ as the values that would drive the \textit{reduced order system} (\ref{rom}) to the setpoint. These values are computed by solving the linear system
\begin{equation}
\begin{bmatrix}
A - I & B \\ TH & 0
\end{bmatrix}\begin{bmatrix}
\bar{x}_{\infty} \\ \bar{u}_{\infty}
\end{bmatrix} = 
\begin{bmatrix}
0 \\ r
\end{bmatrix}.
\end{equation}

Figure \ref{numericalZ} shows the results of using the ROMPC scheme presented in this paper, where the setpoint is shown in black. In this first plot, the proposed method is compared against the na\"ive method for a case with \textit{no disturbances}. As expected, the proposed method converges to offset-free tracking but the na\"ive method does not. 

In the remaining plots only the proposed ROMPC method is demonstrated, but with different (bounded) disturbance sequences. Additionally, the set $\{r\} \oplus \mathcal{R}$ is plotted, starting at the corresponding time $k_r + \tau_{ss}$ as per Theorem \ref{Rtheorem}. In blue we show a simulation where the disturbances are drawn from a (zero-mean) uniform distribution. The simulation in green features (constant) disturbances that lie on the boundary of the bounding polytopes $\mathcal{W}$, $\mathcal{V}$. Finally, in red, we show a case where the disturbances follow one of the worst-case feasible sequences determined when the set $\mathcal{R}$ was computed. As expected, we see that in each case the tracking variable converges to the set $\{r\} \oplus \mathcal{R}$.

\begin{figure*}[ht]
    \centering
    \begin{subfigure}[t]{.26\textwidth}
        \includegraphics[width=\textwidth]{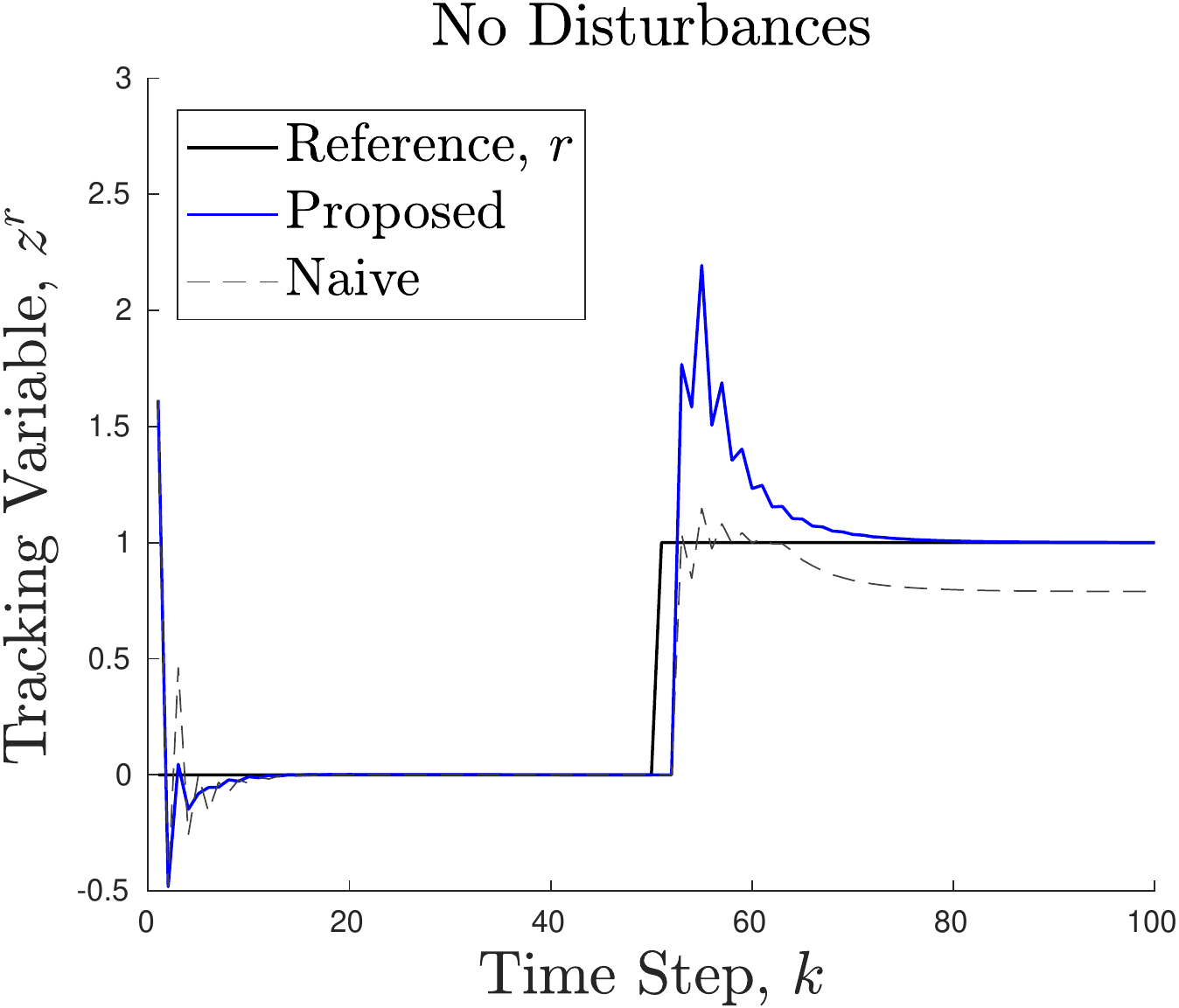}
    \end{subfigure}
    \begin{subfigure}[t]{.24\textwidth}
        \includegraphics[width=\textwidth]{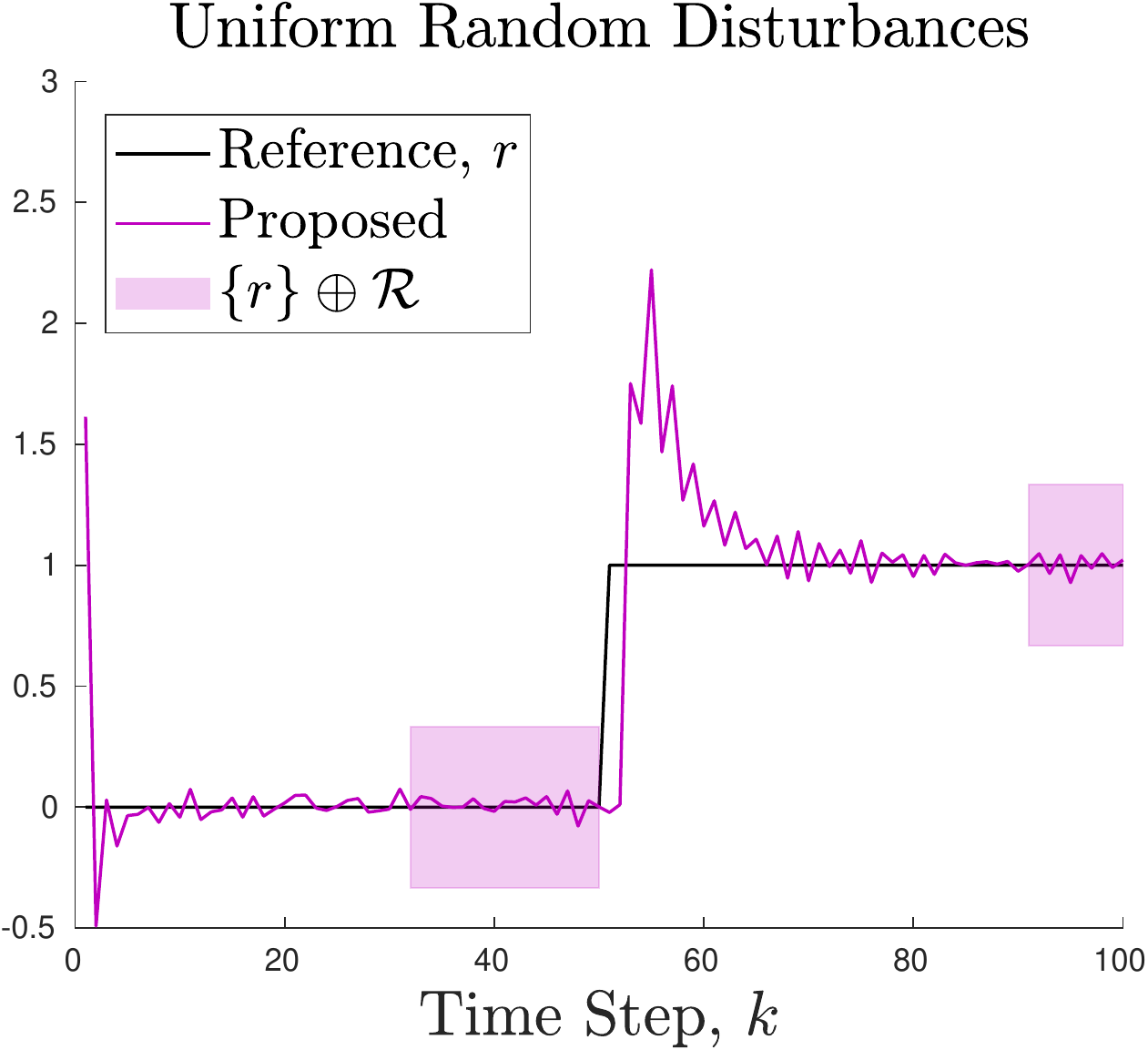}
    \end{subfigure}
    \begin{subfigure}[t]{.24\textwidth}
        \includegraphics[width=\textwidth]{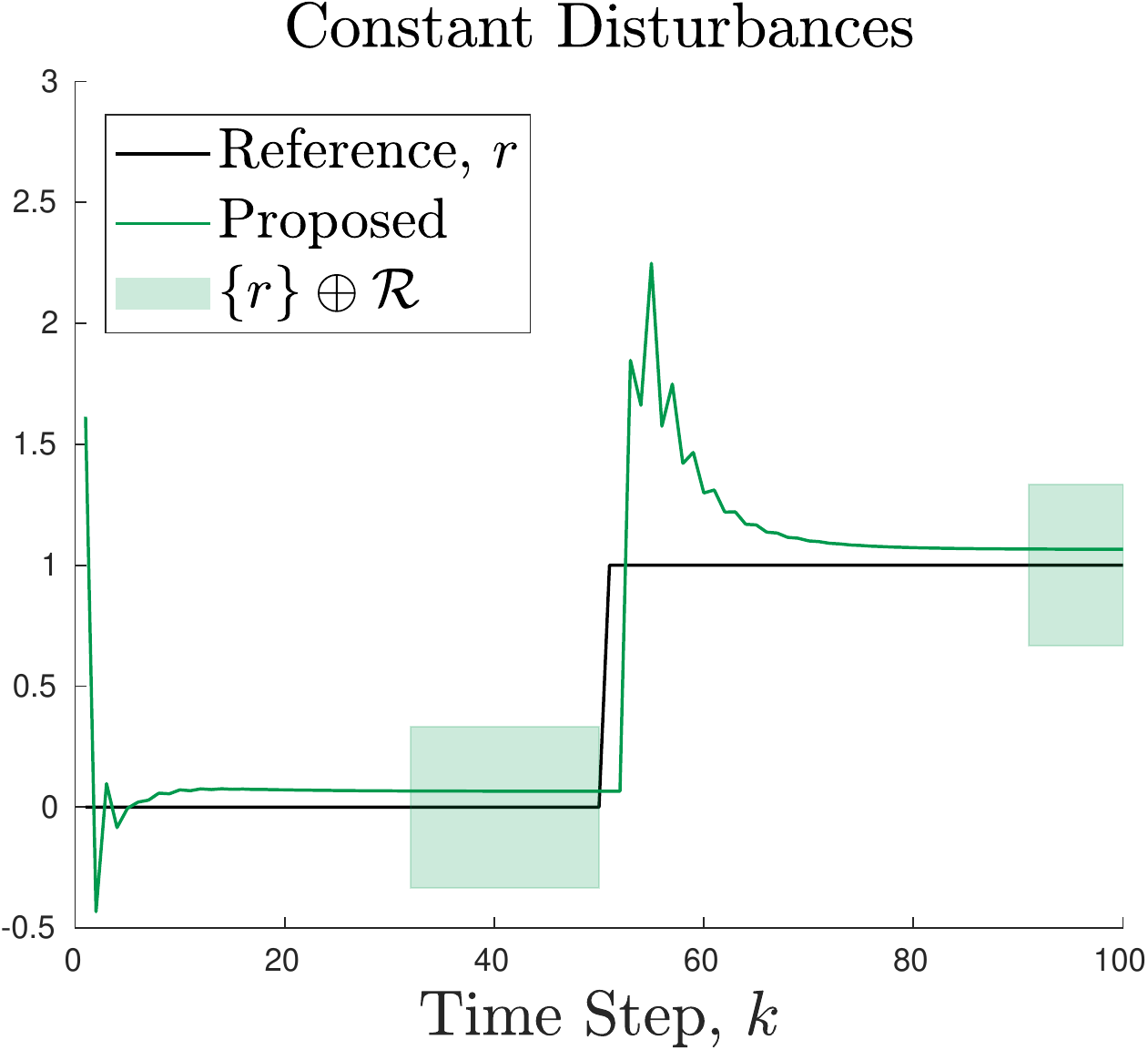}
    \end{subfigure}
    \begin{subfigure}[t]{.24\textwidth}
        \includegraphics[width=\textwidth]{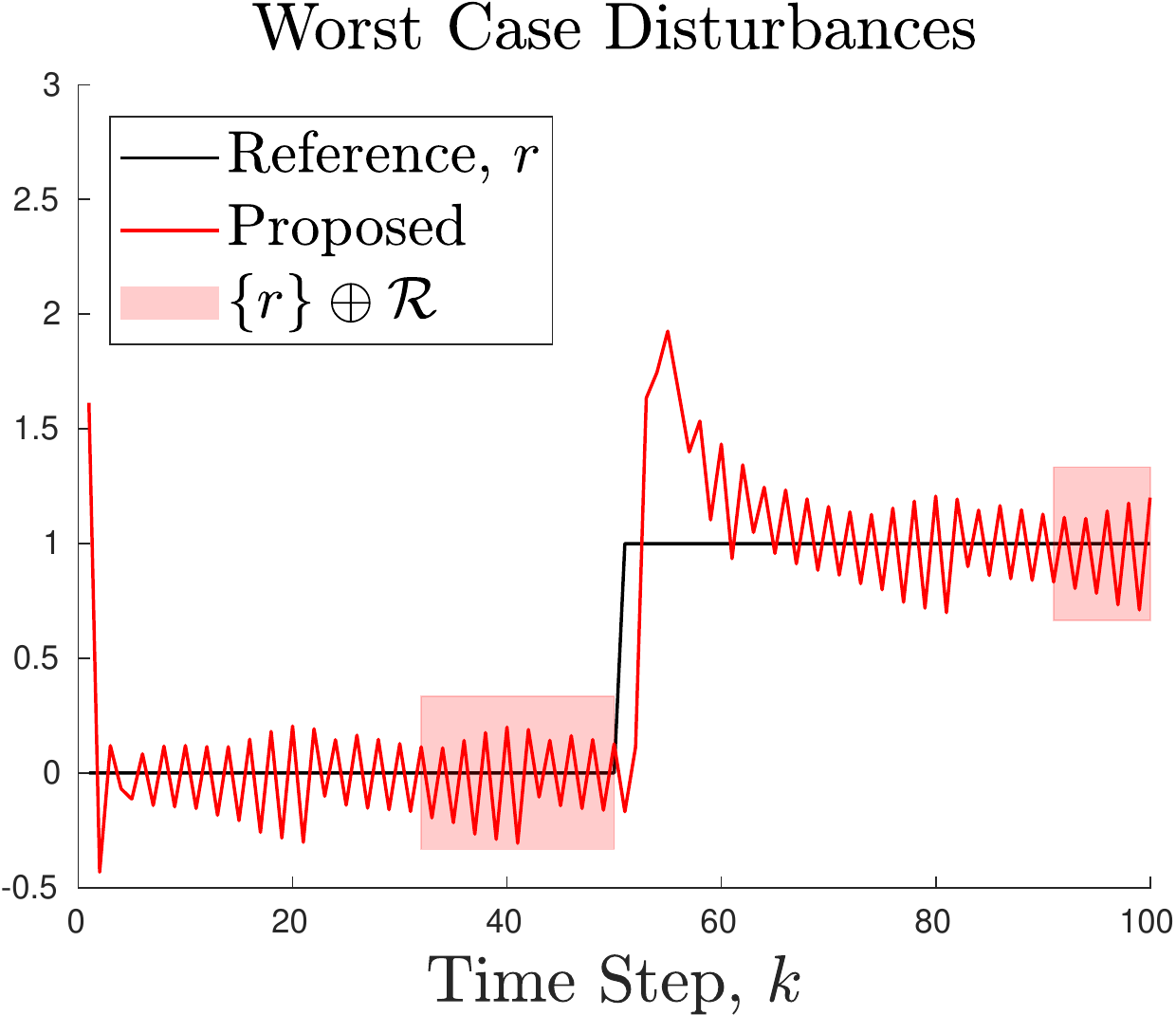}
    \end{subfigure}
    \caption{Setpoint tracking performance demonstrated on a synthetic system discussed in Section \ref{sec:numericalex}, under varying types of bounded disturbances.}
\label{numericalZ}
\vspace{-.20in}
\end{figure*}

\subsection{Flexible Beam Control} \label{sec:flexArm}
We now implement the proposed ROMPC scheme to track a vertical position setpoint with the endpoint of a flexible beam. The flexible beam is modeled by finite elements, as discussed in \cite{JunkinsKim1993, MartinsMohamedEtAl2003}. The beam model includes four nodes and is assumed to be attached to a rigid hub. A torque input to the hub controls the angle of the hub. There is no damping of the motion of the hub, but the flexible modes of the beam are damped. Additionally, a uniform load is applied across the length of the beam. The model's physical parameters represent the Sheffield flexible manipulator described in \cite{MartinsMohamedEtAl2003}.

The resulting model has $n^f = 18$ states, including $2$ integrator modes. The reduced order model is computed using balanced truncation and has dimension $n = 8$. The proposed ROMPC scheme is then applied to track a setpoint with the endpoint of the beam. The results are shown in Figure \ref{flexArmZ}. We see that the proposed method successfully tracks the setpoint under disturbances drawn from a (zero-mean) uniform distribution. However, using the na\"ive approach induces a non-negligible tracking error.
\begin{figure}[ht]
  \centering
  \includegraphics[width=0.33\textwidth]{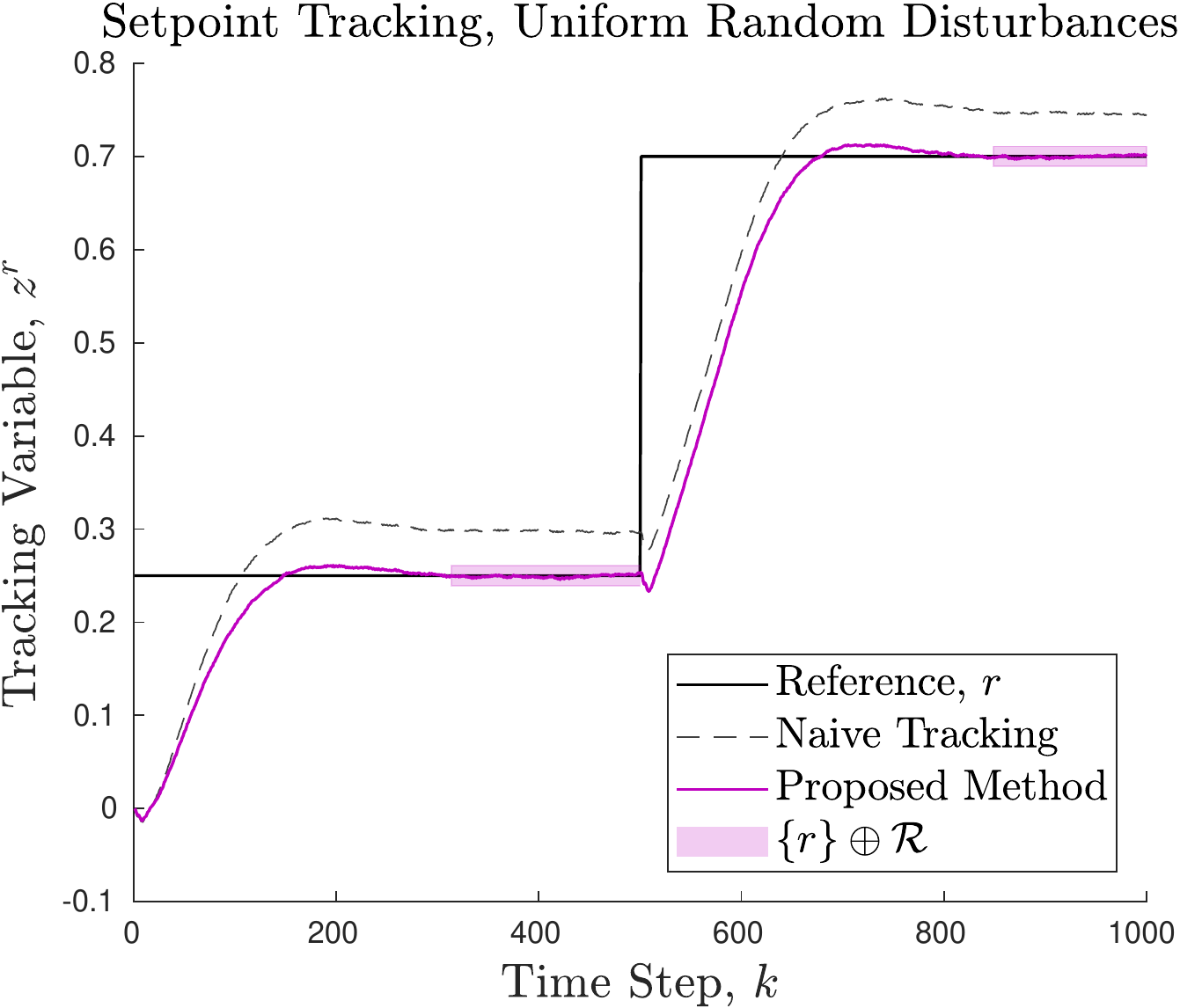}
  \caption{Simulation of the flexible beam tracking problem discussed in Section \ref{sec:flexArm}.}
  \label{flexArmZ}
  \vspace{-.20in}
\end{figure}

\section{Conclusion}\label{sec:conclusion}

In this paper a \textit{reduced order} model predictive control (ROMPC) method is proposed that enables setpoint tracking while robustly guaranteeing constraint satisfaction. Setpoint tracking is accomplished through the design of the ROMPC cost function, and constraint satisfaction is guaranteed by tightening the true constraints using error bounds computed offline. Additionally we provided a design methodology to ensure overall stability and convergence for the algorithm. Finally, the method was validated on a synthetic example and an example inspired by flexible structure control. 

{\em Future Work:} An interesting extension to this work is to reformulate the ROMPC objective function to handle cases where the setpoint cannot be feasibly tracked, as in \cite{AlvaradoLimonEtAl2007}. Further, we have noted that for some designs of the reduced order models, it may not be possible to find non-empty tightened constraint sets $\bar{\mathcal{Z}}$, $\bar{\mathcal{U}}$. Therefore characterizing the properties of full order and reduced order models that lead to non-empty $\bar{\mathcal{Z}}$, $\bar{\mathcal{U}}$ is crucial. This would also enable the reduced order model and the control scheme to be co-designed in a better way. Finally, we look forward to applying our approach to the control of real world robotic systems, including those with nonlinear dynamics.

\bibliographystyle{IEEEtran}
\bibliography{./main,./ASL_papers}

\newcommand{\noopsort}[1]{} \newcommand{\printfirst}[2]{#1}
  \newcommand{\singleletter}[1]{#1} \newcommand{\switchargs}[2]{#2#1}
\begin{thebibliography}{10}
\providecommand{\url}[1]{#1}
\csname url@samestyle\endcsname
\providecommand{\newblock}{\relax}
\providecommand{\bibinfo}[2]{#2}
\providecommand{\BIBentrySTDinterwordspacing}{\spaceskip=0pt\relax}
\providecommand{\BIBentryALTinterwordstretchfactor}{4}
\providecommand{\BIBentryALTinterwordspacing}{\spaceskip=\fontdimen2\font plus
\BIBentryALTinterwordstretchfactor\fontdimen3\font minus
  \fontdimen4\font\relax}
\providecommand{\BIBforeignlanguage}[2]{{%
\expandafter\ifx\csname l@#1\endcsname\relax
\typeout{** WARNING: IEEEtran.bst: No hyphenation pattern has been}%
\typeout{** loaded for the language `#1'. Using the pattern for}%
\typeout{** the default language instead.}%
\else
\language=\csname l@#1\endcsname
\fi
#2}}
\providecommand{\BIBdecl}{\relax}
\BIBdecl

\bibitem{ErenPrachEtAl2017}
U.~Eren, A.~Prach, B.~B. Ko{\c{c}}er, S.~V. Rakovi{\'c}, E.~Kayacan, and
  B.~A{\c{c}}ikmese, ``Model predictive control in aerospace systems: Current
  state and opportunities,'' \emph{{AIAA Journal of Guidance, Control, and
  Dynamics}}, vol.~40, no.~7, pp. 1541--1566, 2017.

\bibitem{PoignetGautier2000}
P.~Poignet and M.~Gautier, ``Nonlinear model predictive control of a robot
  manipulator,'' in \emph{{Int.\ Workshop on Advanced Motion Control}}, 2000.

\bibitem{BealGerdes2013}
C.~E. Beal and J.~C. Gerdes, ``Model predictive control for vehicle
  stabilization at the limits of handling,'' \emph{{IEEE Transactions on
  Control Systems Technology}}, vol.~21, no.~4, pp. 1258--1269, 2013.

\bibitem{GouryDuriez2018}
O.~Goury and C.~Duriez, ``Fast, generic and reliable control and simulation of
  soft robots using model order reduction,'' \emph{{IEEE Transactions on
  Robotics}}, 2018.

\bibitem{RaoPanEtAl1990}
S.~Rao, T.~Pan, and V.~Venkayya, ``Modeling, control, and design of flexible
  structures: A survey,'' \emph{{Applied Mechanics Reviews}}, vol.~43, no.~5,
  pp. 99--117, 1990.

\bibitem{AmsallemDeolalikarEtAl2013}
D.~Amsallem, S.~Deolalikar, F.~Gurrola, and C.~Farhat, ``Model predictive
  control under coupled fluid-structure constraints using a database of
  reduced-order models on a tablet,'' in \emph{{AIAA Aviation Technology,
  Integration, and Operations (ATIO) Conference}}, 2013.

\bibitem{ArdakaniBridges2011}
H.~A. Ardakani and T.~J. Bridges, ``Shallow-water sloshing in vessels
  undergoing prescribed rigid-body motion in three dimensions,'' \emph{{Journal
  of Fluid Mechanics}}, vol. 667, pp. 474--519, 2011.

\bibitem{Antoulas2005}
A.~Antoulas, \emph{Approximation of Large-Scale Dynamical Systems}.\hskip 1em
  plus 0.5em minus 0.4em\relax {SIAM}, 2005.

\bibitem{RawlingsMayne2017}
J.~B. Rawlings and D.~Q. Mayne, \emph{Model Predictive Control: Theory,
  Computation, and Design}.\hskip 1em plus 0.5em minus 0.4em\relax {Nob Hill
  Publishing}, 2017.

\bibitem{MaederBorrelliEtAl2009}
U.~Maeder, F.~Borrelli, and M.~Morari, ``Linear offset-free model predictive
  control,'' \emph{{Automatica}}, vol.~45, no.~10, pp. 2214--2222, 2009.

\bibitem{AlvaradoLimonEtAl2007}
I.~Alvarado, D.~Limon, T.~Alamo, and E.~F. Camacho, ``Output feedback robust
  tube based {MPC} for tracking of piece-wise constant references,'' in
  \emph{{Proc.\ IEEE Conf.\ on Decision and Control}}, 2007.

\bibitem{AstridHuismanEtAl2002}
P.~Astrid, L.~Huisman, S.~Weiland, and A.~C. P.~M. Backx, ``Reduction and
  predictive control design for a computational fluid dynamics model,'' in
  \emph{{Proc.\ IEEE Conf.\ on Decision and Control}}, 2002.

\bibitem{HovlandWillcoxEtAl2006}
S.~Hovland, K.~Willcox, and J.~T. Gravdahl, ``{MPC} for large-scale systems via
  model reduction and multiparametric quadratic programming,'' in \emph{{Proc.\
  IEEE Conf.\ on Decision and Control}}, 2006.

\bibitem{HovlandGravdahlEtAl2008}
S.~Hovland, J.~T. Gravdahl, and K.~E. Willcox, ``Explicit model predictive
  control for large-scale systems via model reduction,'' \emph{{AIAA Journal of
  Guidance, Control, and Dynamics}}, vol.~31, no.~4, pp. 918--926, 2008.

\bibitem{HovlandLovaasEtAl2008}
S.~Hovland, C.~Lovaas, J.~T. Gravdahl, and G.~C. Goodwin, ``Stability of model
  predictive control based on reduced-order models,'' in \emph{{Proc.\ IEEE
  Conf.\ on Decision and Control}}, 2008.

\bibitem{SopasakisBernardiniEtAl2013}
P.~Sopasakis, D.~Bernardini, and A.~Bemporad, ``Constrained model predictive
  control based on reduced-order models,'' in \emph{{Proc.\ IEEE Conf.\ on
  Decision and Control}}, 2013.

\bibitem{LoehningRebleEtAl2014}
M.~L{\"o}hning, M.~Reble, J.~Hasenauer, S.~Yu, and F.~Allg{\"o}wer, ``Model
  predictive control using reduced order models: Guaranteed stability for
  constrained linear systems,'' \emph{{Journal of Process Control}}, vol.~24,
  no.~11, pp. 1647--1659, 2014.

\bibitem{KoegelFindeisen2015b}
M.~K{\"o}gel and R.~Findeisen, ``Robust output feedback model predictive
  control using reduced order models,'' \emph{{IFAC-Papers Online}}, vol.~48,
  no.~8, pp. 1008--1014, 2015.

\bibitem{MayneRakovicEtAl2006}
D.~Q. Mayne, S.~V. Rakovi{\'c}, R.~Findeisen, and F.~Allg{\"o}wer, ``Robust
  output feedback model predictive control of constrained linear systems,''
  \emph{{Automatica}}, vol.~42, no.~7, pp. 1217--1222, 2006.

\bibitem{RakovicKerriganEtAl2005}
S.~V. Rakovic, E.~C. Kerrigan, K.~I. Kouramas, and D.~Q. Mayne, ``Invariant
  approximations of the minimal robust positively invariant set,'' \emph{{IEEE
  Transactions on Automatic Control}}, vol.~50, no.~3, pp. 406--410, 2005.

\bibitem{JunkinsKim1993}
J.~Junkins and Y.~Kim, \emph{Introduction to Dynamics and Control of Flexible
  Structures}.\hskip 1em plus 0.5em minus 0.4em\relax {American Institute of
  Aeronautics and Astronautics}, 1993, ch. Mathematical Models of Flexible
  Structures, pp. 139--234.

\bibitem{MartinsMohamedEtAl2003}
J.~M. Martins, Z.~Mohamed, M.~O. Tokhi, J.~S{\'a}~da Costa, and M.~A. Botto,
  ``Approaches for dynamic modelling of flexible manipulator systems,''
  \emph{{IEE Proceedings - Control Theory and Applications}}, vol. 150, no.~4,
  pp. 401--411, 2003.

\end{thebibliography}

\end{document}